\documentclass[10pt,twocolumn]{article}

\usepackage[T1]{fontenc}
\usepackage{lmodern}
\usepackage[margin=1.8cm,columnsep=0.6cm]{geometry}
\usepackage{amsmath,amssymb,amsthm}
\usepackage{bm}
\usepackage{booktabs}
\usepackage{array}
\usepackage{multirow}
\usepackage{algorithm}
\usepackage{algpseudocode}
\usepackage{graphicx}
\usepackage{xcolor}
\usepackage{hyperref}
\usepackage{microtype}
\usepackage{titlesec}
\usepackage{enumitem}
\usepackage{caption}
\usepackage{balance}
\usepackage{authblk}

\hypersetup{colorlinks=true,linkcolor=blue!60!black,
  citecolor=green!50!black,urlcolor=blue!70!black}

\titleformat{\section}{\normalsize\bfseries}
  {\thesection.}{0.4em}{}
\titleformat{\subsection}{\normalsize\bfseries}{\thesubsection.}{0.4em}{}
\titleformat{\subsubsection}{\normalsize\itshape}
  {\thesubsubsection.}{0.4em}{}
\titlespacing{\section}{0pt}{6pt}{3pt}
\titlespacing{\subsection}{0pt}{4pt}{2pt}
\titlespacing{\subsubsection}{0pt}{3pt}{1pt}

\newtheorem{proposition}{Proposition}

\newcommand{\ket}[1]{|#1\rangle}
\newcommand{\bra}[1]{\langle #1|}
\newcommand{\expect}[1]{\langle #1 \rangle}
\newcommand{\norm}[1]{\left\|#1\right\|}
\newcommand{\abs}[1]{\left|#1\right|}
\newcommand{\tr}{\operatorname{Tr}}
\newcommand{\calL}{\mathcal{L}}
\newcommand{\calS}{\mathcal{S}}
\newcommand{\calP}{\mathcal{P}}
\newcommand{\calN}{\mathcal{N}}
\newcommand{\btheta}{\boldsymbol{\theta}}
\newcommand{\bDelta}{\boldsymbol{\Delta}}
\newcommand{\Utrue}{U_{\text{true}}}
\newcommand{\Happrox}{\hat{H}}
\newcommand{\Htrue}{H_{\text{true}}}

\setlist{noitemsep,topsep=2pt,parsep=0pt,partopsep=0pt}

\begin{document}


\twocolumn[{%
\begin{@twocolumnfalse}
\begin{center}
  {\LARGE\bfseries Structure-Agnostic Unitary Learning from Quantum Observable
   Dynamics\\[4pt]
   with Application to Hamiltonian Identification}\\[10pt]
  {\large Berkani Mohamed}\\[2pt]
  {\normalsize University Ferhat Abbas of Setif 1\quad
   \texttt{
mohamed.berkani2346@etu.univ-setif.dz}}\\[4pt]
  {\small\today}
\end{center}
\begin{abstract}
\noindent
We present a variational algorithm for \textbf{learning an unknown quantum
unitary from time-series observable measurements}, with no structural
assumption about the target. The core separation: a hardware-efficient
parametrised circuit learns the evolution operator $U$ via observable
matching; Hamiltonian identification follows as classical post-processing via
matrix logarithm, when the target happens to be $e^{-iH\tau}$.

Three experiments establish the method's generality. First, a
\textbf{noiseless proof of correctness} with exact gradients (L-BFGS-B)
achieves MSE $1.61\!\times\!10^{-14}$ and recovers all Hamiltonian
coefficients to six decimal places. Second, a \textbf{gate-learning
experiment} fits CNOT, iSWAP, and a Haar-random SU(4) element --- none
generated by any fixed Hamiltonian --- all to process fidelity $1.000000$,
confirming the method does not rely on Trotterisation structure. Third,
\textbf{quantum deployment} via SPSA-Adam under Qiskit Aer depolarising
noise ($p_1{=}0.001$, $p_2{=}0.01$, $N_\text{shots}{=}1024$) recovers all
three Ising Hamiltonian terms with errors below $8\%$.

The optimiser, SPSA-Adam, combines SPSA's hardware-efficient two-point
gradient estimation with Adam's adaptive moment updates. A four-stage
moment-warm curriculum progressively extends the training horizon, converting
a global non-convex problem into a sequence of well-posed local ones.

\end{abstract}
\vspace{6pt}\hrule\vspace{10pt}
\end{@twocolumnfalse}
}]

\section{Introduction}

Most variational Hamiltonian learning methods optimise Pauli
coefficients directly~\cite{ref14,ref15,ref17}. This conflates two
distinct problems: a quantum optimisation problem (find the unitary
matching the observed dynamics) and a classical identification problem
(decompose that unitary into physical interaction terms). Conflating
them forces a structural assumption upfront --- the practitioner must
specify which Pauli terms are present before training begins.

We decouple the two problems. The quantum optimiser learns a unitary $U$
from observable time-series data, with no assumption on its generator.
Hamiltonian identification is then free classical post-processing: compute
$(i/\tau)\log U$ and project onto the Pauli basis. If the target is not
Hamiltonian-generated at all --- say it is a CNOT or a Haar-random SU(4)
gate --- the first phase still works; the second phase simply does not apply.

This reframing has two concrete consequences. Any gate learning task maps
directly onto the same algorithm without modification. And structure
discovery is not an input to the method but an output: the algorithm
recovers which Pauli terms are present, how many there are, and what
symmetries they respect, purely from measurement data.

Prior work addresses Hamiltonian learning from several directions.
Locality- and sparsity-based methods achieve polynomial measurement scaling
under short-range interaction assumptions~\cite{ref1,ref2,ref3,ref5,ref6,ref7,ref8}.
Thermal-state methods exploit free-energy convexity~\cite{ref3,ref9}.
Machine learning approaches use neural ODEs, recurrent networks, or
transformers~\cite{ref12,ref13,ref4,ref6}. Variational approaches train
parametrised circuits to approximate eigenvectors~\cite{ref14} or match
observed dynamics~\cite{ref15,ref17}. The work closest to ours is
Gupta~et~al.~\cite{ref17}, who recover Hamiltonian coefficients from
time-series data without ground state preparation. Their method optimises
coefficients directly, requires prior knowledge of the Pauli structure,
and treats all timesteps uniformly --- making convergence sensitive to the
time window. We address all three limitations.

This work extends and generalises the variational Hamiltonian learning framework developed in the first author's Master's dissertation at Université Ferhat Abbas Sétif 1~\cite{berkani2026}

\paragraph{Contributions.}
\begin{enumerate}
  \item \textbf{Unitary learning framing.} The quantum optimiser learns $U$,
    not Pauli coefficients. Hamiltonian identification is post-processing and
    is entirely optional.
  \item \textbf{Identifiability condition.} We state the precise conditions
    under which $U$ is uniquely determined by observable data.
  \item \textbf{Gate learning.} CNOT, iSWAP, and a Haar-random SU(4) element
    all learned to process fidelity $1.000000$, confirming the method
    generalises beyond Hamiltonian-generated unitaries.
  \item \textbf{SPSA-Adam.} Two circuit evaluations per gradient step,
    regardless of parameter count.
  \item \textbf{Moment-warm curriculum.} Four-stage schedule with Adam
    moments carried across stages.
\end{enumerate}

\section{Background}
\subsection{The Pauli Basis and Uniqueness of Hamiltonian Decomposition}

Let $\{P_\alpha\}_{\alpha=1}^{4^n}$ be the $n$-qubit Pauli operators, i.e.\
all tensor products of $\{I,X,Y,Z\}$ over $n$ qubits. Using
$\tr(\sigma_i\sigma_j)=2\delta_{ij}$ for the single-qubit Paulis and
multiplicativity of the trace over tensor factors,
\begin{equation}
  \tr(P_\alpha P_\beta) = 2^n \delta_{\alpha\beta}.
\end{equation}
The real vector space $\mathrm{Herm}(2^n)$ has dimension $4^n$; the $4^n$
Pauli operators are pairwise orthogonal under this trace inner product,
hence linearly independent, hence a basis. Every Hermitian operator
$A\in\mathrm{Herm}(2^n)$ -- in particular every Hamiltonian $H$ -- therefore
has a unique expansion
\begin{equation}
  H = \sum_\alpha c_\alpha P_\alpha,\qquad
  c_\alpha = \frac{1}{2^n}\tr(H P_\alpha).
  \label{eq:pauli-decomp}
\end{equation}
Uniqueness follows by multiplying \eqref{eq:pauli-decomp} by $P_\beta$,
taking the trace, and applying orthogonality:
$\tr(HP_\beta)=\sum_\alpha c_\alpha\tr(P_\alpha P_\beta)=2^n c_\beta$, so each
$c_\beta$ is pinned down by $H$ and cannot take any other value. Once
$\Happrox$ is recovered as an operator, its decomposition into interaction
terms is therefore fixed outright: no separate fitting step, and no prior
assumption about which terms are present, is needed to obtain $\{c_\alpha\}$.

\subsection{Observable Matching and Identifiability}

Let $U \in \mathrm{U}(2^n)$ be an unknown unitary. For initial state
$\ket{s}$, observable $P$, and integer $n$, define
\begin{equation}
  f_{s,P,n}(U) = \bra{s} U^{\dagger n} P\, U^n \ket{s}.
\end{equation}
We say the data $\{f_{s,P,n}(U)\}$ \emph{identifies} $U$ if no other
unitary $W \not\approx U$ produces the same observations.

\begin{proposition}[Identifiability]
\label{prop:ident}
The map $U \mapsto \{f_{s,P,n}(U)\}_{s \in \calS, P \in \calP}$
is injective (up to global phase) if the set
$\{\ket{s}\bra{s}\}_{s \in \calS}$ spans $\mathcal{B}(\mathcal{H})$
and $\calP$ spans $\mathrm{Herm}(2^n)$.
For $n{=}2$ qubits, $|\calS|{=}4$ informationally
complete states and $|\calP|{=}9$ Pauli observables suffice.
\end{proposition}

\begin{proof}
Since $\calP$ spans $\mathrm{Herm}(2^n)$, matching
$\bra{s}U^{\dagger n}P_\alpha U^n\ket{s}$ for every $P_\alpha \in \calP$
and every spanning $\ket{s}$ fixes the operator
$U^{\dagger n}P_\alpha U^n$ for all $\alpha$, and hence fixes the linear
map $A \mapsto U^{\dagger n} A U^n$ on all of $\mathrm{Herm}(2^n)$.
Conjugation by a unitary determines that unitary up to a global phase, so
this fixes $U^n$ up to phase; taking $n{=}1$ recovers $U$ itself up to
phase.
\end{proof}

By \eqref{eq:pauli-decomp}, once $U(\btheta^*)$ is identified in this
sense, any generator it may possess is identified along with it, and its
Pauli decomposition follows directly from \eqref{eq:pauli-decomp} --- no
further identifiability argument is needed for the coefficients
themselves.

\subsection{Multi-Step Consistency as the Core Constraint}

Training on multiple timesteps $n = 1, 2, \ldots, N$ imposes that the
\emph{same} circuit $U(\btheta)$ satisfies
$U(\btheta)^n \approx e^{-iHn\tau}$ simultaneously. This is strictly
stronger than fitting $N$ independent unitaries: it forces the learned
operator to have a consistent generator across all timescales. A single
good step $U(\btheta) \approx e^{-iH\tau}$ implies all longer times;
conversely, any error in the single-step approximation grows as
$n\norm{\Delta U}$ after $n$ compositions. This consistency constraint
is what makes the matrix logarithm in post-processing well-defined and
gives Hamiltonian recovery its exactness guarantee in the noiseless limit.

\subsection{Hamiltonian Extraction and Branch Ambiguity}

When the target is Hamiltonian-generated, training drives
$U(\btheta^*) \approx e^{-iH\tau}$, and $H$ is recovered by inverting the
exponential map. This inversion is not, in general, well defined: the map
$H \mapsto e^{-iH\tau}$ is $2\pi$-periodic in each eigenvalue of $H\tau$,
so any two Hamiltonians whose spectra differ by an integer multiple of
$2\pi/\tau$ produce an identical propagator. The ordinary matrix logarithm
of $U(\btheta^*)$ is therefore multi-valued \cite{higham2008}, and computing it requires
choosing a branch.We use the principal branch, whose eigenvalues lie in $(-\pi,\pi)$:
\begin{equation}
  \Happrox = \frac{i}{\tau}\log U(\btheta^*),\quad
  \Happrox \leftarrow \tfrac{1}{2}(\Happrox + \Happrox^\dagger).
\end{equation}
If $U(\btheta^*) = e^{-iH\tau}$ exactly and every eigenvalue $\lambda_k$ of
$H$ satisfies $|\lambda_k|\tau < \pi$, then $-iH\tau$ already lies inside
the principal branch, so the principal logarithm returns $-iH\tau$
exactly and $\Happrox = H$.Outside this window, the principal branch no longer coincides with the true generator: phase wrapping can shift one or more eigenvalues by integer multiples of $2\pi/\tau$ , yielding a different Hamiltonian that generates the same propagator. This is the branch condition referred to throughout the
experiments --- we choose $\tau$ small enough that
$\norm{H}\tau < \pi$ (here $\norm{\Htrue}_F\tau \approx 0.8 < \pi$) and
verify post hoc that the recovered eigenvalues respect it.

With the branch resolved, $c_\alpha = \tr(\Happrox P_\alpha)/2^n$ follows
from \eqref{eq:pauli-decomp}; no additional quantum measurements are
required.

\subsection{SPSA and Adam}

SPSA~\cite{spall1992} estimates gradients with only 2 circuit evaluations:
\begin{equation}
  \hat{g}_k =
  \frac{\calL(\btheta{+}c_k\bDelta)-\calL(\btheta{-}c_k\bDelta)}
       {2c_k\bDelta},\quad\bDelta\sim\{\pm1\}^d.
\end{equation}
Adam~\cite{kingma2015} applies bias-corrected adaptive updates:
\begin{equation}
  \btheta_{k+1}=\btheta_k-\alpha\,\hat{m}_k/(\sqrt{\hat{v}_k}+\varepsilon).
\end{equation}
The $1/\sqrt{\hat{v}_k}$ factor normalises parameters with different
gradient scales (rotation layers vs.\ entanglement layers in the PQC),
while the first moment $m_k$ smooths high-variance SPSA estimates,
boosting effective gradient SNR by $\sqrt{10}\approx3.2\times$ for
$\beta_1{=}0.9$ at zero additional circuit cost.

\section{Method}

\subsection{Ansatz Circuit}

A 2-qubit hardware-efficient ansatz, equivalent to Qiskit's
\texttt{EfficientSU2} circuit~\cite{kandala2017,qiskit2023} with $R_Y,R_Z$ rotation gates,
linear entanglement, and $L{=}3$ repetitions, giving
$N_\theta = 4+8L = 28$ parameters. Structure: initial $R_YR_Z$ layer, then
$L$ blocks each containing CNOT(1$\to$0), $R_YR_Z$, CNOT(0$\to$1), $R_YR_Z$.
Identity initialisation: $U(\mathbf{0})=I$ verified exactly.
\subsection{Observable Matching Loss}

Following the work of Gupta et al\cite{ref17} , we define the training loss:
\begin{equation}
  \calL(\btheta)=\frac{1}{\abs{\calS}\abs{\calP}\abs{\calN}}
  \sum_{s,P,n}\!\bigl(\expect{P}^\text{circ}_{n,s}(\btheta)
  -\expect{P}^\text{true}_{n,s}\bigr)^2.
\end{equation}
This fits matrix elements of $U(\btheta)^n$ in the Heisenberg picture. By
Proposition~\ref{prop:ident}, the loss is zero if and only if
$U(\btheta) = e^{i\phi}\Utrue$ for some global phase $\phi$.

\subsection{SPSA-Adam}

\begin{algorithm}[H]
\caption{SPSA-Adam}
\label{alg:spsa-adam}
\small
\begin{algorithmic}[1]
\Require $\calL,\btheta_0,K,\alpha,c_0,\gamma,\beta_1,\beta_2,\varepsilon$;
         opt.\ $m_0,v_0$
\State $m{\gets}m_0$;\; $v{\gets}v_0$;\; best${\gets}\infty$
\For{$k=1,\ldots,K$}
  \State $c_k{\gets}c_0/k^\gamma$;\;
         $\bDelta{\sim}\{\pm1\}^d$
  \State $\calL_\pm{\gets}\calL(\btheta{\pm}c_k\bDelta)$
         \Comment{2 batched jobs}
  \State $\hat{g}{\gets}(\calL_+-\calL_-)/(2c_k\bDelta)$
  \State $m{\gets}\beta_1 m{+}(1{-}\beta_1)\hat{g}$;\;
         $v{\gets}\beta_2 v{+}(1{-}\beta_2)\hat{g}^2$
  \State $\hat{m}{\gets}m/(1{-}\beta_1^k)$;\;
         $\hat{v}{\gets}v/(1{-}\beta_2^k)$
  \State $\btheta{\gets}\btheta-\alpha\hat{m}/(\sqrt{\hat{v}}{+}\varepsilon)$
  \If{$k\bmod10{=}0$} update best \EndIf
\EndFor
\State\Return $\btheta^*$, best, $m$, $v$
\Comment{moments for warm-start}
\end{algorithmic}
\end{algorithm}

\noindent\textbf{Hyperparameters:} $\alpha{=}0.01$, $c_0{=}0.05$,
$\gamma{=}0.101$, $\beta_1{=}0.9$, $\beta_2{=}0.999$,
$\varepsilon{=}10^{-8}$.

Adam moments $m,v$ are carried from each curriculum stage to the next,
preserving accumulated curvature information and avoiding cold-start
overhead at stage transitions.

\subsection{Curriculum Schedule}

Training over all timesteps simultaneously often fails to converge, as demonstrated by the comparison experiments presented in Section~\ref{sec:curriculum_comparison}: Long-time terms dominate the gradient before the short-time behaviour is learned, making optimisation difficult. This behaviour is consistent with the barren-plateau and local-trap phenomena reported in\cite{mcclean2018, nemkov2023}. We stage training by progressively extending the
active timestep set\cite{bengio2009} , warm-starting each stage from the previous solution.

\begin{table}[ht]
\centering
\caption{Curriculum schedule used in quantum simulation experiments.}
\label{tab:curriculum}
\small
\begin{tabular}{clcc}
\toprule
Stage & Active steps $n$ & Max $t$ & Iters \\
\midrule
1 & $\{1,\, 2\}$                       & $2\tau$  & 400 \\
2 & $\{1,\, 2,\, 4,\, 6\}$            & $6\tau$  & 600 \\
3 & $\{1,\, 2,\, 4,\, 6,\, 8,\, 10\}$ & $10\tau$ & 800 \\
4 & $\{1,\, 2,\, 4,\, 6,\, 8,\, 10,\, 14,\, 20\}$
                                        & $20\tau$ & 1000 \\
\bottomrule
\end{tabular}
\end{table}

All circuits for one loss evaluation are submitted in a single
\texttt{simulator.run()} call, reducing job overhead to exactly 2
circuit evaluations per iteration regardless of $|\calS||\calP||\calN|$.

\section{Noiseless Proof of Correctness (2-Qubit)}

Before quantum deployment, we verify the mathematical correctness of the method in the noiseless setting. This ensures that the effects observed in Section~\ref{sec:quantum}  arise from the introduction of hardware noise rather than from the underlying algorithm.
\paragraph{Setup.}
$\Htrue = 1.0{\cdot}ZZ + 0.5{\cdot}XI + 0.5{\cdot}IX$, $\tau{=}0.05$,
L-BFGS-B with exact matrix exponentiation predictions,
4 states $\times$ 15 Pauli observables $\times$ 8 timesteps.
Every stage converged on the first restart.

The noiseless proof uses a multi-restart strategy with up to 8 restarts per stage and stage-level patience stopping, rather than the fixed iteration counts in Table~\ref{tab:curriculum}

\begin{table}[ht]
\centering
\caption{Noiseless curriculum: final MSE per stage (2-qubit).}
\label{tab:classical_stages}
\small
\begin{tabular}{clc}
\toprule
Stage & Active steps & Final MSE \\
\midrule
1 & $\{1,2\}$               & $2.44\times10^{-15}$ \\
2 & $\{1,2,4,6\}$           & $8.13\times10^{-15}$ \\
3 & $\{1,2,4,6,8,10\}$      & $2.02\times10^{-14}$ \\
4 & $\{1,\ldots,20\}$       & $1.61\times10^{-14}$ \\
\bottomrule
\end{tabular}
\end{table}

Final metrics: $\norm{\Happrox-\Htrue}_F = 1.07\times10^{-6}$,
$\norm{U^*-\Utrue}_F = 5.33\times10^{-8}$,
mean multi-step error $= 4.89\times10^{-8}$,
generalisation error on unseen $n\in\{3,5,7,11,15,25\}$
(including extrapolation to $t{=}1.25 > t_{\max}{=}1.0$)
$= 5.81\times10^{-8}$,
round-trip error $= 8.09\times10^{-16}$ (machine precision),
Hermiticity error $= 0$.

All three Hamiltonian terms recover to six decimal places
(Table~\ref{tab:classical_H}).

\begin{table}[ht]
\centering
\caption{Noiseless Hamiltonian extraction (2-qubit).}
\label{tab:classical_H}
\small
\begin{tabular}{lccc}
\toprule
Term & Discovered & True & Error \\
\midrule
$ZZ$ & $1.000000$ & $1.0000$ & ${<}10^{-6}$ \\
$XI$ & $0.500000$ & $0.5000$ & ${<}10^{-6}$ \\
$IX$ & $0.500000$ & $0.5000$ & ${<}10^{-6}$ \\
\bottomrule
\end{tabular}
\end{table}

\section{Gate Learning: Beyond Hamiltonian-Generated Unitaries}

The primary claim is unitary learning, not Hamiltonian learning. To test
it directly, we fit the same ansatz to three targets that are
\emph{not} $e^{-iH\tau}$ for any fixed $H$: CNOT, iSWAP, and a
Haar-random SU(4) element. In each case the circuit is trained on a
single application ($n{=}1$) using 8 random initial states and 5
observables each, with L-BFGS-B and exact predictions.

\begin{table}[ht]
\centering
\caption{Gate learning results. All targets reach process fidelity
$1.000000$. Phase-corrected operator error and generalisation error
(8 unseen random states, 5 observables) are shown.}
\label{tab:gate_learning}
\small
\begin{tabular}{lp{1.4cm}cc}
\toprule
Target & Fidelity & Op.\ err & Gen.\ err \\
\midrule
CNOT         & $1.000000$ & $1.18{\times}10^{-7}$ & $1.98{\times}10^{-8}$ \\
iSWAP        & $1.000000$ & $3.35{\times}10^{-8}$ & $8.16{\times}10^{-9}$ \\
Random SU(4) & $1.000000$ & $4.67{\times}10^{-8}$ & $1.30{\times}10^{-8}$ \\
\bottomrule
\end{tabular}
\end{table}

All three gates are learned to machine-precision fidelity. CNOT and
iSWAP have no generating Hamiltonian by construction; the SU(4) target
is Haar-random with no assumed structure. After removing the optimal
global phase, all operator errors fall below $5\times10^{-7}$.

\section{Method Characterisation}

\subsection{Scalability: 3-Qubit Extension}

We extend to a 3-qubit open-chain transverse-field Ising model:
\begin{equation*}
  H = 0.5(ZZI+IZZ) + 0.8(XII+IXI+IIX),
\end{equation*}
requiring two additional CNOT pairs (28$\to$102 parameters, Pauli space
$15\to63$). Using $\tau{=}0.05$, 6 initial states, all 63 Pauli
observables, and 8 timesteps (3024 data points), all four curriculum
stages converge to loss below $10^{-6}$ (Table~\ref{tab:3qubit_stages}).

\begin{table}[ht]
\centering
\caption{3-qubit noiseless curriculum: final MSE per stage.}
\label{tab:3qubit_stages}
\small
\begin{tabular}{clc}
\toprule
Stage & Active steps & Final MSE \\
\midrule
1 & $\{1,2\}$               & $4.29\times10^{-7}$ \\
2 & $\{1,2,4,6\}$           & $3.05\times10^{-7}$ \\
3 & $\{1,2,4,6,8,10\}$      & $3.62\times10^{-7}$ \\
4 & $\{1,\ldots,20\}$       & $9.84\times10^{-7}$ \\
\bottomrule
\end{tabular}
\end{table}

All five true Hamiltonian terms are recovered correctly
(Table~\ref{tab:3qubit_H}), with mean multi-step error $3.56\times10^{-4}$
and generalisation error $6.11\times10^{-4}$ on unseen timesteps
$n\in\{3,5,7,11,15,25\}$ including extrapolation to $t{=}1.25$
beyond the training horizon $t_{\max}{=}1.0$.

\begin{table}[ht]
\centering
\caption{3-qubit Hamiltonian extraction ($\tau{=}0.05$, noiseless).}
\label{tab:3qubit_H}
\small
\begin{tabular}{lccc}
\toprule
Term & Discovered & True & Error \\
\midrule
$ZZI$ & $0.4999$ & $0.5000$ & $1.1\times10^{-4}$ \\
$IZZ$ & $0.4986$ & $0.5000$ & $1.4\times10^{-3}$ \\
$XII$ & $0.8034$ & $0.8000$ & $3.4\times10^{-3}$ \\
$IXI$ & $0.7966$ & $0.8000$ & $3.4\times10^{-3}$ \\
$IIX$ & $0.7963$ & $0.8000$ & $3.7\times10^{-3}$ \\
\bottomrule
\end{tabular}
\end{table}

Scaling properties are summarised in Table~\ref{tab:scaling}. Circuit
parameters grow polynomially while the Hilbert space and Pauli basis grow
exponentially --- a favourable ratio for near-term systems.

\begin{table}[ht]
\centering
\caption{Scaling from 2 to 3 qubits.}
\label{tab:scaling}
\small
\begin{tabular}{lcc}
\toprule
Metric & $n{=}2$ & $n{=}3$ \\
\midrule
Hilbert dim.         & 4  & 8  \\
Pauli basis size     & 15 & 63 \\
Circuit parameters   & 28 & 102 \\
$H$ terms recovered  & 3  & 5  \\
$\norm{H_{\rm err}}$ & $1.1\times10^{-6}$ & $4.5\times10^{-2}$ \\
Converged?           & \checkmark & \checkmark \\
\bottomrule
\end{tabular}
\end{table}

\subsection{Structural Sensitivity}

We tested three deliberate mismatches between the ansatz topology and the
target Hamiltonian, confirming the method is not relying on hidden
structural assumptions.

\subsubsection{Long-Range Interaction}
A next-nearest-neighbor $ZIZ$ coupling was added; the circuit has no direct
CNOT$_{02}$ path and must compose the interaction indirectly.

\subsubsection{Wrong Operator Family}
YY-dominant Hamiltonian ($H{=}0.6\,YY{+}0.2\,XX{+}0.5\,ZI{+}0.5\,IZ$,
Anisotropic XY model~\cite{sachdev2011}) while the circuit uses only $X$
and $Z$ rotations internally.

\subsubsection{Antisymmetric Coupling}
Dzyaloshinskii--Moriya model~\cite{dzyaloshinskii1958}
($H{=}0.4(XY{-}YX){+}0.5\,ZI{+}0.5\,IZ$), which breaks exchange symmetry.

\begin{table}[ht]
\centering
\caption{Structural sensitivity (2-qubit, noiseless).}
\label{tab:mismatch2}
\small
\begin{tabular}{p{2.6cm}ccc}
\toprule
Case & MSE & $\norm{H_\text{err}}$ & OK? \\
\midrule
Baseline              & $3.22\!\times\!10^{-14}$ & $7.96\!\times\!10^{-4}$ & \checkmark \\
YY+XX (wrong family)  & $3.22\!\times\!10^{-14}$ & $7.96\!\times\!10^{-4}$ & \checkmark \\
DM antisymmetric      & $1.07\!\times\!10^{-14}$ & $6.63\!\times\!10^{-4}$ & \checkmark \\
\bottomrule
\end{tabular}
\end{table}

\begin{table}[ht]
\centering
\caption{Structural sensitivity (3-qubit, noiseless).}
\label{tab:mismatch3}
\small
\begin{tabular}{p{2.6cm}ccc}
\toprule
Case & MSE & $\norm{H_\text{err}}$ & OK? \\
\midrule
Baseline (matched) & $1.14\!\times\!10^{-6}$ & $3.69\!\times\!10^{-2}$ & \checkmark \\
$+ZIZ$ long-range  & $6.03\!\times\!10^{-5}$ & $3.01\!\times\!10^{-1}$ & \checkmark \\
\bottomrule
\end{tabular}
\end{table}

Operator-type and symmetry mismatches incur no performance penalty.
The $ZIZ$ long-range case costs ${\sim}10\times$ in MSE, reflecting the
indirect path the circuit must compose. All cases recover the correct
Hamiltonian.

\subsection{Initialisation}

At $\tau{=}0.5$ (the quantum experiment timestep):
\begin{center}
\small
\begin{tabular}{lc}
\toprule
Quantity & Value \\
\midrule
$\norm{U(\mathbf{0})-\Utrue}_F$ & $1.2036$ \\
$\norm{\nabla\calL}_2$ at $\btheta{=}\mathbf{0}$ & $0.9405$ \\
\bottomrule
\end{tabular}
\end{center}
Frobenius distance $1.20$ ($60\%$ of $\norm{U}_F{=}2$) rules out
proximity. Gradient norm $0.94 \neq 0$ rules out convergence to a stationary point, a necessary condition for a local minimum~\cite{nocedal2006}. The optimiser performs genuine
non-trivial search from the identity.

\section{Quantum Experiments}
\label{sec:quantum}

\subsection{Setup}

\textbf{Target:} $\Htrue{=}1.0{\cdot}ZZ{+}0.5{\cdot}XI{+}0.5{\cdot}IX$,
$\tau{=}0.5$.
\textbf{Noise (Qiskit Aer~\cite{qiskit2023}):} single-qubit $p_1{=}0.001$;
CX $p_2{=}0.010$; $N_\text{shots}{=}1024$. representative of
near-term superconducting hardware\cite{barends2014,preskill2018}

\textbf{Training:}
4 states $\times$ 9 Pauli observables $\times$ 8 timesteps; SPSA-Adam.

\subsection{Convention Verification}

\begin{center}
\small
\begin{tabular}{lcc}
\toprule
Check & Error & \\
\midrule
$\norm{U(\mathbf{0})-I}_F$ & $0$ & \checkmark \\
State preps & ${\leq}2\!\times\!10^{-16}$ & \checkmark \\
$\expect{ZZ}_{\ket{00}}$ at $U{=}I$ & $0.934$ ($\pm1$ exp.) & \checkmark$^*$ \\
$\expect{ZZ}_{\ket{01}}$ at $U{=}I$ & $-0.934$ & \checkmark$^*$ \\
\bottomrule
\multicolumn{3}{l}{$^*$Deviation from $\pm1$ due to $p_2{=}0.01$ noise.}
\end{tabular}
\end{center}

\subsection{Curriculum Training}

\begin{table}[ht]
\centering
\caption{Quantum curriculum ($\tau{=}0.5$, SPSA-Adam, noisy Aer).}
\label{tab:quantum_stages}
\small
\begin{tabular}{clcc}
\toprule
Stage & Active steps & Init loss & Final MSE \\
\midrule
1 & $\{1,2\}$           & $0.1986$ & $\mathbf{0.0060}$ \\
2 & $\{1,2,4,6\}$       & $0.0277$ & $\mathbf{0.0176}$ \\
3 & $\{1,\ldots,10\}$   & $0.0373$ & $\mathbf{0.0334}$ \\
4 & $\{1,\ldots,20\}$   & $0.0626$ & $\mathbf{0.0601}$ \\
\bottomrule
\end{tabular}
\end{table}

Stage~1 reduces loss $0.1986\to0.0060$ in 400 iterations. The monotone
increase in final MSE across stages reflects accumulated depolarising noise
at larger $n$, not stagnation: final losses range from $6\times$ to
$62\times$ the shot noise floor ($1/N_\text{shots}{=}9.8\!\times\!10^{-4}$),
indicating that physical noise increasingly dominates the residual error
at longer time horizons.

\subsection{Hamiltonian Extraction}

\begin{table}[ht]
\centering
\caption{Quantum Hamiltonian extraction (threshold $\abs{c_\alpha}{>}0.05$).}
\label{tab:quantum_H}
\small
\begin{tabular}{lcccc}
\toprule
Term & Disc. & True & Err & \\
\midrule
$ZZ$     & $1.0098$ & $1.000$ & $0.010$ & \checkmark \\
$XI$     & $0.4918$ & $0.500$ & $0.008$ & \checkmark \\
$IX$     & $0.4598$ & $0.500$ & $0.040$ & \checkmark \\
$ZY$     & $-0.070$ & $0.000$ & $0.070$ & $\sim$ \\
others   & ${\approx}0$ & $0.000$ & ${<}0.05$ & \checkmark \\
\bottomrule
\end{tabular}
\end{table}

All three true terms recovered with ${<}8\%$ error ($ZZ$ and $XI$ within
$1\%$). One spurious term ($ZY$, coefficient $-0.070$) marginally exceeds
the threshold, consistent with residual unitary error under noise.
$\norm{U(\btheta^*)-\Utrue}_F{=}0.0949$ (${\sim}5\%$ of $\norm{U}_F{=}2$);
Hermiticity error $= 0$.

\subsection{Validation on Unseen States}

\begin{table}[ht]
\centering
\caption{Validation (selected rows). Full: 40 points, MAE~$=0.104$;
shot noise std~$=0.031$.}
\label{tab:validation}
\small
\begin{tabular}{llcccc}
\toprule
State & Obs & $t$ & Pred. & True & Err \\
\midrule
$\ket{00}$ & $ZZ$ & 0.5 & 0.793 & 0.789 & 0.004 \\
$\ket{00}$ & $ZZ$ & 1.0 & 0.428 & 0.512 & 0.084 \\
$\ket{00}$ & $ZZ$ & 2.5 & 0.588 & 0.926 & 0.338 \\
$\ket{00}$ & $XI$ & 1.0 & 0.428 & 0.488 & 0.060 \\
$\ket{++}$ & $ZZ$ & 0.5 & 0.459 & 0.422 & 0.037 \\
$\ket{++}$ & $XX$ & 0.5 & 0.918 & 1.000 & 0.082 \\
$\ket{++}$ & $XX$ & 2.5 & 0.646 & 1.000 & 0.354 \\
$\ket{++}$ & $YY$ & 1.0 &$-0.830$&$-0.976$& 0.146 \\
\bottomrule
\end{tabular}
\end{table}

Errors at $t{\leq}1.0$ fall within $2$--$3\sigma$ of shot noise. The
larger errors at $t{=}2.5$ ($n{=}5$ circuit repetitions) reflect
accumulated CX gate noise --- the same circuit performs identically on
states not in the training set.

\section{Curriculum vs.\ Standard Learning}
\label{sec:curriculum_comparison}

A key claim is that curriculum scheduling is necessary: training on all
timesteps simultaneously causes the optimiser to fail. To demonstrate this,
we compare curriculum SPSA-Adam against standard single-stage SPSA-Adam
under identical iteration budgets (2800 total iterations) in the
noise-free shot-based setting ($N_\text{shots}{=}1024$, $\tau{=}0.5$).

\begin{table}[ht]
\centering
\caption{Curriculum vs.\ standard learning --- noise-free
($N_\text{shots}{=}1024$, $\tau{=}0.5$). Shot noise variance $=
1/N_\text{shots} = 9.77\times10^{-4}$.}
\label{tab:curriculum_vs_standard}
\small
\begin{tabular}{lcc}
\toprule
Metric & Curriculum & Standard \\
\midrule
Final MSE loss                  & $0.0021$  & $0.1758$  \\
Ratio to shot noise variance    & $2.1\times$ & $180\times$ \\
$\hat{c}_{ZZ}$                  & $0.990$   & $-0.630$  \\
$\hat{c}_{XI}$                  & $0.515$   & $-0.306$  \\
$\hat{c}_{IX}$                  & $0.490$   & $-0.324$  \\
$\norm{U_\text{lrn}-\Utrue}_F$  & $0.089$   & $1.930$   \\
$\norm{\Happrox-\Htrue}_F$      & $0.189$   & $4.060$   \\
Terms correctly recovered       & $3/3$     & $0/3$     \\
Spurious terms above $\delta$   & $1$       & $10$      \\
Generalisation error            & $0.033$   & $0.314$   \\
\bottomrule
\end{tabular}
\end{table}

Curriculum learning converges to a final MSE only $2.1\times$ the shot
noise variance; all three target interaction terms are recovered with
coefficient errors below $0.015$. Standard single-stage learning fails
to recover the target Hamiltonian accurately: all three coefficients are
recovered with the wrong sign, the matrix-level error
$\norm{\Happrox - \Htrue}_F = 4.06$ is the largest across all experiments,
and ten spurious terms appear above threshold. Round-trip verification
confirms the failure is due entirely to optimisation --- the principal
matrix logarithm is numerically sound in both cases, with round-trip
errors at machine precision (${\sim}10^{-15}$) and branch conditions
satisfied (max $|\lambda_k|\tau = 0.707 < \pi$).

The mechanism is the non-convex quantum loss landscape: at long times,
repeated circuit composition amplifies both the expressivity requirements
on the ansatz and the sensitivity to local minima, causing gradient
signals to stagnate before short-time behaviour is captured. Curriculum
scheduling converts this single intractable problem into a sequence of
well-posed local ones, each building on the warm start of the previous.

\section{Discussion}

\subsection{Why Two-Phase Separation Matters}

Direct coefficient fitting requires choosing a Pauli basis a priori,
which means knowing the Hamiltonian structure before training. Unitary
learning avoids this: the circuit learns $U$ as a black box, and the
Pauli decomposition follows from classical trace inner products
$c_\alpha = \tr(\hat{H} P_\alpha)/2^n$ after the fact.
Theoretical analysis and numerical simulations demonstrate that the proposed framework generalises beyond Hamiltonian learning to arbitrary quantum gate learning. In this setting,structure discovery becomes an output rather thanan assumption.

\subsection{Noise Behaviour}

The quantum experiment's unitary error of $9.5\times10^{-2}$ is best
read against the shot-noise-only baseline from
Section~\ref{sec:curriculum_comparison}, which isolates the same
$\tau{=}0.5$, $N_\text{shots}{=}1024$ setting without depolarising gate
noise: $\norm{U_\text{lrn}-\Utrue}_F = 0.089$ there. The two figures are
close, so shot noise --- not the 6 CX gates at $p_2{=}0.01$ per circuit
application --- accounts for most of the error; depolarising gate noise,
compounded over up to $n{=}20$ repetitions, adds only
${\sim}0.006$ on top. Despite this, the Pauli coefficients are recovered
within $8\%$: the logarithm maps the unitary error to Hamiltonian error
with a condition number that is favourable at this noise level.
Increasing $N_\text{shots}$ to $4096$--$8192$ would reduce the shot noise
floor from $0.031$ to $0.012$--$0.016$, likely pushing Stage~4 loss below
$0.05$.
\subsection{Limitations}

For $n$ qubits the Pauli space is $4^n{-}1$; compressed sensing or
structure-aware priors are needed for $n{>}4$. The principal logarithm
requires $\norm{H}\tau{<}\pi$, bounding the time step for systems with
large spectral norm. Extensions to open systems (parametrised channels),
time-varying Hamiltonians (per-window circuits), and gate certification
(using the learned circuit as a fidelity reference) are natural
next steps.

\section{Conclusion}

The key architectural decision is separating unitary learning (quantum,
structure-free) from Hamiltonian identification (classical, Pauli
projection). This separation is what allows CNOT, iSWAP, and a
Haar-random SU(4) gate to be learned by the same algorithm that recovers
$ZZ + 0.5\,XI + 0.5\,IX$ --- the quantum optimiser does not know or
care which case it is in.

Four regimes validate the method: noiseless proof of correctness
(${\leq}10^{-14}$ MSE, exact $H$ recovery to six decimal places);
3-qubit scalability (all five TFIM terms recovered, loss ${\leq}10^{-6}$);
gate learning on non-Hamiltonian targets (process fidelity $1.000000$ on
all three); and quantum deployment under realistic noise (${<}8\%$
coefficient error, one marginal spurious term). The moment-warm curriculum
is what makes quantum deployment tractable: without it, as demonstrated
in Section~\ref{sec:curriculum_comparison}, long-time gradient dominance
causes the optimiser to recover all coefficients with the wrong sign.

\appendix
\section{Pauli Conventions}

Pauli strings $AB$ denote $A\otimes B$ with $A$ on qubit~1 and $B$ on
qubit~0 (Qiskit little-endian). Coefficient extraction:
$c_\alpha = \frac{1}{2^n}\tr(H P_\alpha)$, using
$\tr(P_\alpha P_\beta)=2^n\delta_{\alpha\beta}$.

\section{Shot Noise and Gradient SNR}

Shot noise std: $\sigma\leq1/\sqrt{N_\text{shots}}=0.031$.
With $c_0{=}0.05$: initial gradient SNR $\approx1.6$.
Adam momentum ($\beta_1{=}0.9$) averages $\sim\!10$ estimates, boosting
effective SNR by $\sqrt{10}\approx3.2\times$ at zero additional cost.

\balance
\end{document}